\documentclass[letterpaper, 10pt,  conference]{ieeeconf}

\IEEEoverridecommandlockouts
\usepackage{inputenc}
\usepackage[noadjust]{cite}
\usepackage{nicematrix}
\usepackage{float}
\usepackage{longtable}
\usepackage{algorithm}
\usepackage{algpseudocode}
\usepackage{amsmath,amssymb,amsfonts}
\usepackage[english]{babel}

\usepackage{ulem}
\usepackage{tabto}
\usepackage{graphicx}
\usepackage{textcomp}
\usepackage{xcolor}
\usepackage{enumerate}
\usepackage[bb=dsserif]{mathalpha}
\usepackage{bm}
\usepackage{soul}
\usepackage{accents}

\usepackage{mathrsfs}
\usepackage{comment}
\usepackage{caption}
\usepackage{subcaption}
\usepackage{pgfplots}
\usepackage{adjustbox}
\usepackage{balance}
\usepackage{csquotes}
\usepackage{tikz}
\usetikzlibrary{arrows.meta}

\newcommand{\R}{\mathbb{R}}
\newcommand{\Rn}{\R^n}

\newcommand{\G}{\mathcal{G}}
\newcommand{\fd}{\tilde{f}}

\def\BibTeX{{\rm B\kern-.05em{\sc i\kern-.025em b}\kern-.08em
    T\kern-.1667em\lower.7ex\hbox{E}\kern-.125emX}}
    \newtheorem{theorem}{Theorem}
\newtheorem{lemma}{Lemma} 
\newtheorem{problem}{Problem} 
 
\newtheorem{assumption}{Assumption}
\newtheorem{definition}{Definition}

\allowdisplaybreaks

\begin{document}

\title{Optimal Pinning Control for Synchronization over Temporal Networks\\
}
\author{Aandrew Baggio S and Rachel Kalpana Kalaimani
\thanks{The authors are with the Department of Electrical Engineering, Indian Institute of Technology Madras, Chennai 600036 India.}
\thanks{
        {Email: \small ee20d067@smail.iitm.ac.in and  rachel@ee.iitm.ac.in}}%
}


\maketitle

\begin{abstract} 
In this paper, we address the finite time synchronization of a network of dynamical systems with time-varying interactions modelled using temporal networks. We synchronize a few nodes initially using external control inputs. These nodes are termed as {\it pinning} nodes. The other nodes are synchronized by interacting with the pinning nodes and with each other. 
We first provide sufficient conditions for the network to be synchronized. Then we formulate an optimization problem to minimize the number of pinning nodes for synchronizing the entire network. Finally, we address the problem of maximizing the number of synchronized nodes when there are constraints on the number of nodes that could be pinned. We show that this problem belongs to the class of NP-hard problems and propose a greedy heuristic. We illustrate the results using numerical simulations.
\end{abstract}


Modern applications involve a large number of dynamical entities interacting with each other. Synchronization of all these entities to some desired system trajectory is of interest in certain applications such as power grids with multiple generators \cite{zha2009}, secure communication \cite{liu2021edge}, harmonic oscillation generation, regulatory mechanisms in biological processes such as synchronous beat of heart cells \cite{pes1975}, opinion dynamics, etc. Synchronization in networked dynamical systems is well studied in literature (\cite{zha2009,li2003,sorrentino2007controllability}). Paper \cite{tang2014synchronization} provides a survey on synchronization in complex networks and its applications.

In the context of large networks, pinning synchronization is of interest. Here, only a fraction of the systems in the network are pinned or in other words, influenced by an external control input. The effect of this input propagates to  other systems by interactions that exist among the systems. The pinning method has been of interest to the research community for the past two decades\cite{wang2002pinning,su2012decentralized,li2004pinning,lu2010global,sorrentino2007controllability}. 
The literature discusses two major approaches to selecting the pinning nodes : Master Stability Function-based approaches (\cite{zhan2007chaos,sorrentino2007controllability,porfiri2008criteria}) and the Lyapunov Function based approaches (\cite{song2009pinning,liuzza2020pinning,della2023nonlinear,montenbruck2015practical}).

The topology of the network which depicts the interactions among systems plays an important role in the synchronization of these networked systems. Literature predominantly assumes these interactions to be static. However, the interactions in complex networks, for example, regulatory networks from biological sciences, epidemic spreading and social networks are known to be time-varying in nature. 
The time-varying interactions influence the sequence of synchronization of the systems in the complex network and therefore should be addressed appropriately.
Hence, we propose to address the synchronization problem for non-static interactions by using  temporal networks. Temporal networks have been used to model time-varying interactions in literature \cite{val2022,masuda2013temporal,ghosh2022synchronized}.
For time-varying interactions, a switched system framework has been adopted in \cite{zha2009}. Here synchronization under arbitrary switching topology is discussed. Then a switching sequence, choosing from a collection of pre-defined topologies, is designed for synchronization. Simultaneous triangularizability of the connection matrices is required for this case. In the temporal network framework, the network topology is not limited to a finite set. Also in our model, we assume that the external input is used only at the beginning for the pinning nodes and then the external input is removed. The systems in the network get synchronized by interactions with their neighbours. Additionally, we design control laws for synchronization in finite time and hence at the end of the temporal network evolution, the network would be synchronized. 

Optimization problems have also been addressed in the pinning control synchronization literature. The pinning nodes are chosen to optimize the  convergence rate in \cite{jafarizadeh2022pinning}, the number of synchronized nodes in \cite{delellis2018partial} and \cite{iudice2022bounded} and stabilizability in \cite{lu2010global}. While all the above papers address the Optimization problem in static networks, we focus on temporal networks.
Synchronization problems are also  addressed when there are time delays \cite{zhou2021pinning}, quantized output control with actuator faults \cite{yan2020} and  noise perturbations \cite{wang2022pinning}.

We consider a network of dynamical systems each with non-linear dynamics with time-varying interactions that are modelled using temporal networks.  We are interested in finite time synchronization of this network using pinning control.
We summarize our contribution as follows.
\begin{itemize}
    \item We formulate sufficient conditions for the networked system to be synchronized in finite time for a given set of pinned nodes. 
    \item Find the minimum number of nodes to be pinned to achieve synchronization of the entire network. This is formulated as a Linear Programming problem. 
    \item Given some constraints on the number of nodes to be pinned, find the maximum number of nodes that could be synchronized. We show that this problem is NP-hard and hence propose a greedy heuristic. We analyze the performance of the heuristic using numerical simulations. 
\end{itemize}

\section{Preliminaries}
Notations : $\mathbb{R}^{n}$ denotes all vectors with real entries of size $n$, $\mathbb{R}^{n \times m}$ denotes the set of all matrices with real entries of size $n \times m$. For $A \in \mathbb{R}^{n \times m}$, $A_{ij}$ denotes the $(i,j)^{th}$ element of $A$ and $A^T$ denotes the transpose of $A$. If $A=A^T \in \mathbb{R}^{n \times n}$, then $\lambda_{max}(A)$ denotes the largest eigenvalue. 
 $\mathbb{0}_{n}$ is a vector of all zeros in $\Rn$. 
 $\mathbb{R}_{+}$ is the set of all positive real numbers. 
 
Let $G=(V,E)$ denote a digraph (directed graph), where $V=\{1, \dots,n\}$ is the set of its vertices  and $E$ is the set of its edges. $E$ is a set of tuples of the form $(i,j)$ which indicates the presence of an edge from  node $j$ to  node $i$. $A$ is the adjacency matrix of $G$ where $A_{ij}> 0$ if $(i,j)\in E$ and $A_{ij}=0$ otherwise. 
$D$ is a diagonal matrix with its $i^{th}$ diagonal element being the in-degree of the $i^{th}$ node in $G$. Then
$L:=D-A$ is defined as the Laplacian matrix associated with $G$. 
A path in a digraph is an ordered sequence of edges, $\{e_{1},\dots,e_{p}\}$, where the node at which an edge ends is the node from which the edge next in the sequence starts (i.e) $e_{k}=(i,j),\text{ }e_{k+1}=(l,i)$ $\forall$ $k\in\{1 ,\dots, , p-1\}$. 
A directed cycle in a digraph is a directed path where the node from where the first edge starts is the same node in which the last edge ends. (i.e) $e_{p}=(i,j)$ and $e_{1}=(j,k)$. 
Digraphs without any directed cycles are termed as Directed Acyclic Graphs (DAGs). A DAG with a finite number of nodes will have at least one node with no incoming edges. We term such nodes as root nodes. 

Let $V$ be a finite set. Consider a set function $f:2^V\rightarrow \R$ that maps any subset of $V$ to $\R$. 
The function $f$ being submodular is defined below. 
\begin{definition}
    A set function $f:2^V\rightarrow \R$ is submodular if for all subsets $P\subseteq Q \subseteq V$ and all set members $s\notin Q$, it holds that
\begin{equation}
   f(P\cup \{s\})-f(P) \geq f(Q\cup \{s\})-f(Q)\label{eq:submod}
\end{equation}
\end{definition}
Another equivalent condition for submodularity is : for any two subsets $P$ and $Q$ of $V$, it holds that
\begin{equation}
  f(P)+f(Q)\geq f(P\cup Q)+f(P\cap Q) \label{eq:submod1}  
\end{equation}

\section{Problem Formulation}\label{sec:prob_form}

We consider N identical dynamical systems  each with n states. They interact with each other and these interactions vary with time. Each interaction persists for a fixed time period. The behaviour of such systems can be modelled using temporal networks defined below.
\begin{definition}
    A temporal network is an ordered sequence of graphs denoted by $\{G_k(V,E_k)\}$, with $k=1$ , . . , $T$ on a set of vertices $V$. Each $v_i\in V$ refers to  a dynamical system. Each element in $E_k$ indicates a directed interaction between two nodes in the time interval $[k\tau$, $(k+1)\tau)$.
\end{definition}
Each graph $G_k$ which persists for time $\tau$ is referred to as a {\it snapshot}. Since each element of $V$ refers to a system. We use nodes and systems interchangeably. 
We assume that we cannot control the sequence of the snapshots. 
Our objective in this case is to synchronize all the systems to some desired trajectory at the end of the temporal network evolution, i.e., at time $\tau T$
 \begin{definition}
    The state vector of a system $i$, denoted by $x_i$, is said to be synchronized to a trajectory $s(t)$ if the following holds
    \begin{align*}
    \lim_{t\rightarrow \infty}||x_{i}(t)-s(t)||_{2} \rightarrow 0
    \end{align*}
\end{definition}
The above definition refers to asymptotic synchronization. We focus on finite time synchronization. 
\begin{definition}
    The states of a system $i$, denoted by $x_i$, is said to be synchronized to a trajectory $s(t)$ in finite time if there exist some $t_{f}>0$ such that the following holds
    \begin{align*}
    ||x_{i}(t)-s(t)||_{2} &= 0 \text{ }\forall\text{ }t\geq t_{f}
    \end{align*}
\end{definition}

Literature on synchronization over static networks assumes the following standard model.
 \begin{equation}
  \dot{x}_{i}(t) = f(x_{i}(t))-c\sum_{j=1}^{N} L_{ij}x_{i}(t)\label{eq:sysmodel}
    \end{equation}
 where $f :\mathbb{R}^{n}\rightarrow \mathbb{R}^{n}$ represents the dynamics of each system, $c\geq0$ is the coupling strength and $L$ represents the Laplacian matrix of the graph depicting the interactions. In the above model, note that the interactions across agents are linear. In \cite{chen2021finite}, the interactions across agents are modified to be non-linear to achieve  finite time synchronization. Inspired from the same, we propose the following dynamics for each agent.  Let ${x}_{i}(t)\in\Rn$ denote the state vector of each agent and 
 $X(t)\in \mathbb{R}^{nN} :=\begin{bmatrix}x_{1}^{T}(t)&\cdots&x_{n}^{T}(t)\end{bmatrix}^{T}$.
  \begin{equation}
  \dot{x}_{i}(t) = f(x_{i}(t))-cg_i^k(X(t)),\,  k\tau\leq t < (k+1)\tau \label{eq:sysmodel - unpinned}
\end{equation}
 where $c\geq 0$ is the coupling strength and 
 
\small{
\begin{align*}
g_{i}^k(X(t)) = &\text{ }\begin{cases}
            \dfrac{\sum_{j=1}^{N}L_{ij}(k)x_{j}(t)}{||\sum_{j=1}^{N}L_{ij}(k)x_{j}(t)||_{2}}, &||\sum_{j=1}^{N}L_{ij}(k)x_{j}(t)||_{2}\neq 0\\
            0, &||\sum_{j=1}^{N}L_{ij}(k)x_{j}(t)||_{2}=0
        \end{cases}, &\nonumber\\
\end{align*}}
\normalsize
$L(k)$ is the Laplacian matrix associated with $G_k$.

 We assume the control model as follows. Initially a few nodes, $S_0\subset V$ will be synchronized to a desired trajectory $s(t)$ by using an external input. This can be achieved in finite time, for an appropriate choice of input gain, as discussed in some literature on Sliding mode Control \cite{yu2017sliding,edwards1998sliding,shtessel2014sliding}. This is elaborated in the Appendix in Section \ref{sec:appendix}. We term the initial set of nodes synchronized by an external input as {\it pinning nodes}. 
 Then, the question that is of interest to us is whether the remaining nodes are also synchronized to $s(t)$ by interacting with other nodes. These time-varying interactions  are depicted by the temporal network $G_k$, for $k=1,\dots T$. 
Let $S_k$ denote the set of all nodes that are synchronized at time $k\tau$. $S_k$ depends on $S_{k-1}$, $G_k$ and $c$. We assume we can choose a high enough value for $c$ such that in finite time $\tau$, depending on $S_{k-1}$ and $G_k$ all possible nodes that can be synchronized will be in $S_k$. A guideline for choosing such a $c$ is provided in Appendix. We next state some of our assumptions on the model. First, we assume the dynamics of all the systems to satisfy the QUAD property defined as follows.
\begin{definition}
A function $f:\mathbb{R}^{n}\rightarrow \mathbb{R}^{n}$ is QUAD($\Delta$ , $\omega$) if it satisfies the following inequality for all $x,y \in \mathbb{R}^{n}$
\begin{align}
&(x-y)^{T}(f(x)-f(y))\leq(x-y)^{T}\left(\Delta-\omega I\right)(x-y)\label{eq:QUAD Condtition}\\
&\Delta\in \mathbb{R}^{n\times n} \text{ is a symmetric matrix and } \omega\in\mathbb{R}_{+}\cup 0 \text{ is a }\nonumber\\
&\text{ non - negative scalar)}\nonumber
\end{align}
\end{definition}

\begin{assumption}
$f$ in \eqref{eq:sysmodel} is QUAD ($\Delta$ , $\omega$)\label{eq:QUAD Assumption}
\end{assumption}
QUAD assumption is commonly used in the literature (\cite{liu2008boundedness,delellis2009novel,delellis2010quad}) when Lyapunov function based approaches are used to derive sufficiency conditions for synchronization.\\ 
The next assumption is on the graph $G_k$ for each snapshot.
\begin{assumption}
    $G_{k}$ is a directed acyclic graph $\forall$ $k=0\dots T$\label{eq:Acyclicity Assumption}
\end{assumption}This assumption helps in  analysing the synchronization of a networked system by exploiting the hierarchy that exists in DAGs. This will be explained in detail in the upcoming section. 
The first problem that we address is to identify conditions for synchronization, given set of pinning nodes. 
Recall that $S_{k}$ depends upon $G_k$ and $S_{k-1}$. We require $S_{T}=V$.
\begin{problem}
Given a set $S_0\subset V$, provide the  conditions for $S_T=V$.
\end{problem}
The next problem is about optimizing the set of pinning nodes.
\begin{problem}
Minimize $S_0$ such that $S_T=V$
\end{problem}
In certain scenarios, the above optimization problem would result in a set $S_0$ with a large cardinality. Practically it might not be feasible to pin a large number of nodes. Hence we address the question of maximizing the number of synchronized nodes when there are constraints on the number of nodes that could be pinned initially. 
\begin{problem}
Maximize $|S_T|$ such that $|S_0|\leq q$
\end{problem}
\section{Main Results}
In this section, we address the three problems mentioned in the problem formulation. We first explain about the graph structure for each $G_k$ and then about the construction of a fused graph $\G$ formed using all the graphs $G_1,G_2,\dots G_T$ in the temporal network. 

In a time interval $[k\tau, (k+1)\tau)$, the root nodes in $G_k$ retain their status (synchronized or unsynchronized) from the time $k\tau$ as there are no incoming edges to these nodes in this duration $[k\tau, (k+1)\tau)$. In order to capture this, we construct a fused graph $\G$ as explained below. 
We stack all the graphs one below the other as shown in Figure \ref{fig:Supergraph Example} on the left side. There is a temporal network on 3 nodes. The nodes which are solid are the root nodes. The edges in each graph $G_k$ depict the interactions that occur in time interval $[k\tau,(k+1)\tau)$. Then we draw edges to each of the root nodes, starting from the same node in the previous time interval. This is illustrated in the right side of Figure \ref{fig:Supergraph Example} 

\begin{figure}[H]
\centering
\begin{tikzpicture}[scale=0.6]
    \node[text width=0.5cm] at (2,0.5) 
    {$\boldsymbol{G_{1}}$};
    \node[text width=0.5cm] at (2,-2.5) 
    {$\boldsymbol{G_{2}}$};
    \node[text width=0.5cm] at (2,-5.5) 
    {$\boldsymbol{G_{3}}$};
\begin{scope}[every node/.style={circle,ultra thick,minimum size=13pt,inner sep=0pt,draw=blue!100,draw}]
    \node (1) at (0,1) {\bf 1};
    \node (2) at (-0.866,-0.5) {\bf 2};
    \node[fill=black,text=white] (3) at (0.866,-0.5)  {\bf 3};
    
    \node[fill=black,text=white] (4) at (0,-2) {\bf 1};
    \node[fill=black,text=white] (5) at (-0.866,-3.5) {\bf 2};
    \node (6) at (0.866,-3.5) {\bf 3};    
    
    \node (7) at (0,-5) {\bf 1};
    \node[fill=black,text=white] (8) at (-0.866,-6.5) {\bf 2};
    \node[fill=black,text=white] (9) at (0.866,-6.5) {\bf 3};

    \node (10) at (4,1) {\bf 1};
    \node (11) at (3.134,-0.5) {\bf 2};
    \node[fill=black,text=white] (12) at (4.866,-0.5) {\bf 3};
    
    \node[fill=black,text=white] (13) at (4,-2) {\bf 1};
    \node[fill=black,text=white] (14) at (3.134,-3.5) {\bf 2};
    \node (15) at (4.866,-3.5) {\bf 3};    
    
    \node (16) at (4,-5) {\bf 1};
    \node[fill=black,text=white] (17) at (3.134,-6.5) {\bf 2};
    \node[fill=black,text=white] (18) at (4.866,-6.5) {\bf 3};
\end{scope}

\draw[blue, very thick,-stealth] (3) -- (1);
\draw[blue, very thick,-stealth] (1) -- (2);

\draw[blue, very thick,-stealth] (5) -- (6);

\draw[blue, very thick,-stealth] (9) -- (7);
\draw[blue, very thick,-stealth] (8) -- (7);

\draw[blue, very thick,-stealth] (12) -- (10);
\draw[blue, very thick,-stealth] (10) -- (11);

\draw[blue, very thick,-stealth] (14) -- (15);

\draw[blue, very thick,-stealth] (18) -- (16);
\draw[blue, very thick,-stealth] (17) -- (16);

\path[->, ultra thick, dashed] (10)  edge   (13);
\path[->, ultra thick, dashed] (11)  edge   (14);

\path[->, ultra thick, dashed] (14)  edge   (17);
\path[->, ultra thick, dashed] (15)  edge   (18);

\end{tikzpicture}
\caption{Fused graph $\G$ of the snapshots $G_{1}$, $G_{2}$ and $G_{3}$. The solid nodes are root nodes.}
\label{fig:Supergraph Example}
\end{figure}
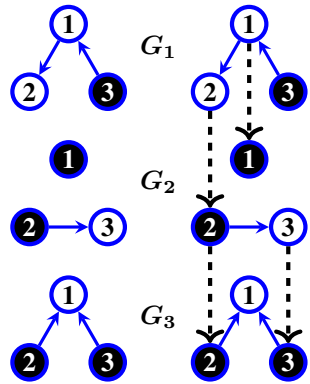

\subsection{Graph Conditions for Synchronizing the temporal network}  
\begin{theorem}\label{thm:cond_sync}
 Consider a network of dynamical systems each following the dynamics as given in \eqref{eq:sysmodel - unpinned},  with interactions depicted by a temporal network $\{G_k\}$, $k=1,\dots, T$. Let $\G$ denote the fused graph obtained from all $G_k$'s.  Given an $S_0$, if the set of all root nodes in $G_{1}$ that have a directed path in $\G$ to any of the root nodes in $G_{T}$ is a subset of $S_{0}$, then $S_{T}=V$, provided the coupling strength $c$ is sufficiently large.
\end{theorem}
\begin{proof}
Let $R_k\subset V$ denote the set of root nodes in the graph $G_k$.  Root nodes do not have any incoming edges and hence they retain their synchronization status within a given snapshot. Note that in any snapshot, a node $i\in V$ is synchronized if all nodes in $R_k$ that have a path to $i$ is synchronized. The other set of root nodes which do not have a path to $i$ do not influence it and hence we do not require them to be synchronized. 
 $S_T=V$, if all the root nodes in $G_T$, i.e., $R_T$ are synchronized.  Hence the nodes in $R_T$ will be synchronized if all nodes in $R_{T-1}$  that have a path to nodes in $R_T$ in $G_{T-1}$ is synchronized. Applying this statement recursively for all snapshots starting from $T$ to $1$, we get that all nodes in $R_1$ that have a path to the nodes in $R_T$ in the fused graph $\G$ should be pinned initially, i.e., they should be a subset of $S_0$ for $S_T=V$.
\end{proof}

\subsection{Optimal pinning nodes}
In the previous section, Theorem \ref{thm:cond_sync} provided a sufficient condition for all the nodes in a temporal network to be synchronized for a given set of initial nodes that were synchronized. In this section we address Problem 2, where we find a minimum set of pinning nodes. This optimization problem can be formulated as an Integer Linear Programming (ILP) problem. Let $s_{i,k}$ denote the status of a system $i$ being  synchronized or not at the time $k\tau$ , i.e., $s_{i,k}=1$ if  synchronized and $s_{i,k}=0$ if unsynchronized. A key concept which is used in formulating the ILP is as follows : A node would be synchronized if all the nodes that have incoming edges to it are  synchronized. Hence, even if one of the incoming edge is unsynchronized, then the  system could become unsynchronized. This is captured by the following condition. 
\begin{equation}
    P_{i,k}s_{i,k}\leq \sum_{j,~ s.t. (i,j)\in E_k} s_{j,k} \label{eq:IP}
\end{equation}
where $P_{i,k}$ is the number of in-neighbours to node $i$ in $G_k$. 

Across two snapshots $k-1$ and $k$, only the root nodes of $G_k$ retain their status. Rest change depending on the interactions depicted by $G_k$. Hence if node $i$ is a root node in $G_k$, then $s_{i,k}=s_{i,k-1}$. We replace the equality by an inequality as follows in order to combine with \eqref{eq:IP}
\begin{equation}\label{eq:trans_cond}
   s_{i,k}\leq s_{i,k-1} 
\end{equation}
We explain shortly that this change will not affect our ILP solution.  
Let $L_{\G}$ denote the Laplacian of the fused graph explained in Figure \ref{fig:Supergraph Example}. Then  equations \eqref{eq:IP} and \eqref{eq:trans_cond} can be together written using $L_{\G}$ as follows.
\[L_{\G}s\leq 0_{nT},\,s^T=\begin{bmatrix}s_{1,0}&\dots&s_{N,0}&s_{1,1}&\dots &s_{N,T}\end{bmatrix}^T\]
Since we need all the nodes to be synchronized at time $\tau T$, we add the following constraint:
$\sum_{i=1}^{N}s_{i,T}=N$, $\forall i$
The ILP is as follows.
\begin{equation}
\begin{aligned}
\min & \sum_{i=1}^{N}s_{i,0}\\
& L_{\G}s\leq 0 \\
 & \sum_{i=1}^{N}s_{i,T}=N\\
&s_{i,k}\in\{0,1\}, \forall i,k
\end{aligned}
\label{eq:ILP}
\end{equation}  
In the above problem all nodes in last snapshot is forced to be 1 by the equality constraint. Then the constraint $s_{i,k}\leq s_{i,k-1}$ for $k=T$ forces all nodes in $G_{k-1}$ that are root nodes in $G_k$ to be 1., i.e., $s_{i,T-1}=s_{i,T}=1$ if $i$ is a root node in $G_T$. Hence replacing the equality constraint by an inequality constraint does not affect the solution.
Since the above problem is ILP, we relax the integer constraints and get an LP as follows.
\begin{equation}
\begin{aligned}
\min & \sum_{i=1}^{N}s_{i,0}\\
& L_{\G}s\leq0 \\
 & \sum_{i=1}^{N}s_{i,T}=N,\\
&0\leq s_{i,k}\leq 1, \forall i,k
\end{aligned}
\label{eq:LP}
\end{equation} 
\begin{lemma}
The solution to the LP in \eqref{eq:LP} is the same as ILP in \eqref{eq:ILP} and is also unique
\end{lemma}
\begin{proof}
The equality constraint of the above LP ensures that $s_{i,T}=1$ $\forall$ $i=1,\dots,N$. Then the constraint in \eqref{eq:trans_cond}, ensures that all nodes in $G_{k-1}$ which are root nodes in $G_k$ will have the status as 1. Then because of constraint in \eqref{eq:IP} all incoming neighbours to the above nodes in $G_{k-1}$ will also have status 1. Using the same argument we can conclude that all variables $s_{i,k}$ will take the value either 1 or 0 and not any real value. Hence the ILP solution and LP solution are the same. 

\end{proof}
\subsection{Maximum synchronizing nodes with constraints on pinning nodes}
In this section we address Problem 3. We first prove that this optimization problem belongs to the class of NP hard problems. 
We do this by showing it as a special case of Submodular Cost Submodular Knapsack (SCSK) problem \cite{iyer2013submodular,padmanabhan2023maximizing,iyer2013fast} which is as follows. For a given $p\leq N$ \\
\begin{equation}
\begin{aligned}
\max_{X \subset V} \quad & g(X)\\
\textrm{s.t.} \quad &f(X) \leq p    \\
\end{aligned}
\label{eq:scsk problem}
\end{equation}
where, $g$ and $f$ are submodular functions.
In our case, we want to maximise $|S_T|$ such that $|S_0|\leq q$.
We define a function $f:V\rightarrow \R$ as follows. For a node $v\in V$, $f(v)$ returns the cardinality of the minimum number of nodes that has to be pinned such that $v\in S_T$. Then the optimization problem is as follows. For a given $p\leq N$
\begin{equation}
\begin{aligned}
\max_{X \subset V} \quad & |X|\\
\textrm{s.t.} \quad &f(X) \leq p \\
\end{aligned}
\label{eq:card_cons_scsk}
\end{equation}
\begin{theorem} \label{thm:fsubmod}
    The function $f$ is submodular
\end{theorem}
\begin{proof}
Define a set function  $\tilde{f}:V\rightarrow V$ as follows. For a set $X\subset V$, $\tilde{f}(X)$ returns all nodes that should be pinned for $X\subset S_T$. Note $|\fd(X)|=f(X)$ If two subsets of nodes of $V$, $A$ and $B$ are to be synchronized, the nodes of $\fd(A)$ and $\fd(B)$ need to be pinned together. So, $\fd(A \cup B) = \fd(A) \cup \fd(B)$. 

Let $A,B\in V$ such that $A \cap B = C$ and $\fd(A)\cap \fd(B)=\tilde{C}\neq\emptyset$.
\begin{align}
f(A)+f(B)  = &|\fd(A)|+|\fd(B)|& \nonumber\\
 = &|(\fd(A)-\tilde{C})\cup\tilde{C}|+|(\fd(B)-\tilde{C})\cup\tilde{C}|& \nonumber\\
= &|\fd(A)-\tilde{C}|+|\tilde{C}|+|\fd(B)-\tilde{C}|+|\tilde{C}|& \nonumber\\
f(A)+f(B) = &|\fd(A)-\tilde{C}| + |\fd(B)-\tilde{C}| + 2|\tilde{C}|&\label{eq:LHS Submodularity}
\end{align}
\vspace{-0.5cm}
\begin{equation*}
f(A \cup B) = |(\fd(A)-\tilde{C})\cup(\fd(B)-\tilde{C})\cup \tilde{C})|
\end{equation*}
\begin{equation}
f(A \cup B)+f(A \cap B) = |(\fd(A)-\tilde{C})|+|(\fd(B)-\tilde{C})| +|\tilde{C}|+f(C)\label{eq:RHS Submodularity} 
\end{equation}

For the $f$ to be submodular, $f(A)+f(B)\geq f(A \cup B)+f(A \cap B)$ $\forall$ $A,B\subseteq V$. From \eqref{eq:LHS Submodularity} and \eqref{eq:RHS Submodularity}, we get the condition : 
\begin{equation}
    |\tilde{C}| \geq f(C) = |\fd(C)|
    \label{eq:Submodularity Constraint}
\end{equation}

We now prove that the inequality \eqref{eq:Submodularity Constraint} always holds true.
\begin{align}
|\tilde{C}| & = |\fd(A)\cap\fd(B)|\nonumber\\
& = \big|\fd((A-C)\cup C) \cap \fd((B-C)\cup C) \big|\nonumber\\
    & = \big|\big( \fd(A-C)\cup \fd(C) \big) \cap \big( \fd(B-C)\cup \fd(C) \big) \big|\nonumber\\
|\tilde{C}| & = \big|\fd(C)\cup \big( \fd(A-C) \cap \fd(B-C)\big) \big|&\label{eq:Submodularity Proof RHS}
\end{align}

Since $ \fd(C) \subseteq \fd(C)\cup \big( \fd(A-C) \cap \fd(B-C)\big)$, \eqref{eq:Submodularity Constraint} always holds true. So, $f(A)+f(B)\geq f(A \cup B)+f(A \cap B)$ $\forall$ $A,B\subseteq V$ and hence, we can conclude that the function $f$ is submodular.
\end{proof}
From Theorem \ref{thm:fsubmod}, we can conclude that optimization problem in \eqref{eq:card_cons_scsk} is NP hard and hence Problem 3 is also NP hard. Therefore we propose to use a greedy algorithm as given in Algorithm \ref{alg:Greedy Algorithm}. 
For each $i\in V$, we identify the set of nodes that needs to be pinned such that $i\in S_T$. This can be computed by using the LP formulation in \eqref{eq:LP}. Modify the constraint $\sum_{i=1}^{N}s_{i,T}=N$ as 
$s_{i,T}=1 $ and $ s_{j,T}=0 $ for all other $j\neq i$.
Then we first add the node that requires least number of nodes to be pinned. Then iteratively we add nodes that require least number of new nodes to be pinned as long as the cardinality constraint on the maximum number of nodes to be pinned is satisfied.   

\begin{algorithm}[h]
\caption{Greedy Algorithm}
\begin{algorithmic}
\State $V = \{1$ , . . , $N\}$
\State $X = \emptyset$
\State $i = \emptyset$
\While{$f(X \cup i) \leq p$ and $|X| \leq N$}
    \State $X = X \cup i$
    \State $i = argmin_{j\in V-X}f(X \cup j)$
\EndWhile
\end{algorithmic}
\label{alg:Greedy Algorithm}
\end{algorithm}
 
\subsection{Simulation results}
We considered a network of Van der Pol oscillators and the dynamics of each oscillator is given by the following equation.
\begin{align*}
    \dot{p}&=q\\
    \dot{q}&=(1-p^2)q-p
\end{align*}

We set the desired trajectory to be the limit cycle in Van der Pol oscillator.  We assume the  oscillators are interacting over a temporal network given in Figure \ref{fig:Snapshots of Van der Pol Network}. Following the guidelines in the Appendix, we chose $c=225$.  We solved the LP given in \eqref{eq:LP} and identified that the two nodes $1$ and $2$ should be synchronized for all the entire network to be synchronized. We observe the entire network is synchronized by the end of 5 snapshots in Figure \ref{fig:Van der Pol Network Plot}. 

Next we implemented the greedy algorithm on a few examples to evaluate its performance. We randomly generated 120 temporal networks. The number of nodes in each network varied between 15 to 20. The cardinality constraints were also chosen randomly. 
We observed that the greedy returned the optimum value in all except for two cases. 

\begin{figure}[h]
\begin{subfigure}{0.15\textwidth}\centering
\begin{tikzpicture}[scale=0.7]
\begin{scope}[every node/.style={circle,ultra thick,draw=blue!100,minimum size=13pt,inner sep=0pt,draw}]
    \node (1) at (0,1) {\bf 1};
    \node (2) at (0.9511,0.3090) {\bf 2};
    \node (3) at (0.5878,-0.8090) {\bf 3};
    \node (4) at (-0.5878,-0.8090) {\bf 4};
    \node (5) at (-0.9511,0.3090) {\bf 5};
\end{scope}
\draw[blue, very thick,-stealth] (2) -- (3);
\draw[blue, very thick,-stealth] (5) -- (4);
\end{tikzpicture}
\caption{Snapshot 1}
\end{subfigure}
\begin{subfigure}{0.15\textwidth}\centering
\begin{tikzpicture}[scale=0.7]
\begin{scope}[every node/.style={circle,ultra thick,draw=blue!100,minimum size=13pt,inner sep=0pt,draw}]
    \node (6) at (4,1) {\bf 1};
    \node (7) at (4.9511,0.3090) {\bf 2};
    \node (8) at (4.5878,-0.8090) {\bf 3};
    \node (9) at (3.4122,-0.8090) {\bf 4};
    \node (10) at (3.0489,0.3090) {\bf 5};
\end{scope}
\draw[blue, very thick,-stealth] (9) -- (7);
\draw[blue, very thick,-stealth] (10) -- (9);
\end{tikzpicture}
\caption{Snapshot 2}
\end{subfigure}
\begin{subfigure}{0.15\textwidth}\centering
\begin{tikzpicture}[scale=0.7]
\begin{scope}[every node/.style={circle,ultra thick,draw=blue!100,minimum size=13pt,inner sep=0pt,draw}]
    \node (11) at (0,-2) {\bf 1};
    \node (12) at (0.9511,-2.691) {\bf 2};
    \node (13) at (0.5878,-3.8090) {\bf 3};
    \node (14) at (-0.5878,-3.8090) {\bf 4};
    \node (15) at (-0.9511,-2.691) {\bf 5};
\end{scope}
\draw[blue, very thick,-stealth] (12) -- (14);
\draw[blue, very thick,-stealth] (13) -- (15);
\end{tikzpicture}
\caption{Snapshot 3}
\end{subfigure}
\begin{subfigure}{0.15\textwidth}\centering
\begin{tikzpicture}[scale=0.7]
\begin{scope}[every node/.style={circle,ultra thick,draw=blue!100,minimum size=13pt,inner sep=0pt,draw}]
    \node (16) at (4,-2) {\bf 1};
    \node (17) at (4.9511,-2.691) {\bf 2};
    \node (18) at (4.5878,-3.8090) {\bf 3};
    \node (19) at (3.4122,-3.8090) {\bf 4};
    \node (20) at (3.0489,-2.691) {\bf 5};
\end{scope}
\draw[blue, very thick,-stealth] (17) -- (20);
\draw[blue, very thick,-stealth] (17) -- (19);
\end{tikzpicture}
\caption{Snapshot 4}
\end{subfigure}
\centering
\begin{subfigure}{0.15\textwidth}\centering
\begin{tikzpicture}[scale=0.7]
\begin{scope}[every node/.style={circle,ultra thick,draw=blue!100,minimum size=13pt,inner sep=0pt,draw}]
    \node (21) at (2,-5) {\bf 1};
    \node (22) at (2.9511,-5.691) {\bf 2};
    \node (23) at (2.5878,-6.8090) {\bf 3};
    \node (24) at (1.4122,-6.8090) {\bf 4};
    \node (25) at (1.0489,-5.691) {\bf 5};
\end{scope}
\draw[blue, very thick,-stealth] (21) -- (24);
\draw[blue, very thick,-stealth] (23) -- (25);
\draw[blue, very thick,-stealth] (23) -- (22);
\end{tikzpicture}
\caption{Snapshot 5}
\end{subfigure}
\caption{Snapshot representation of the Van der Pol Network}
\label{fig:Snapshots of Van der Pol Network}
\end{figure}
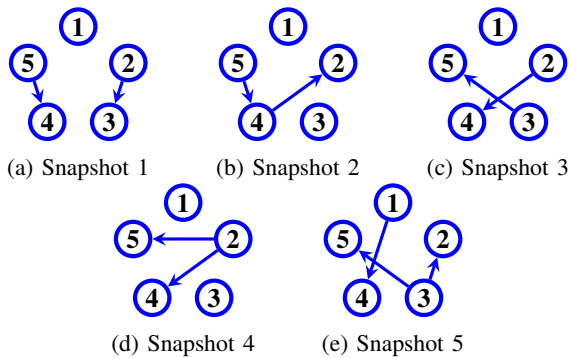

\begin{figure}[h!]
        \begin{subfigure}{0.22\textwidth}
            \includegraphics[width=\textwidth]{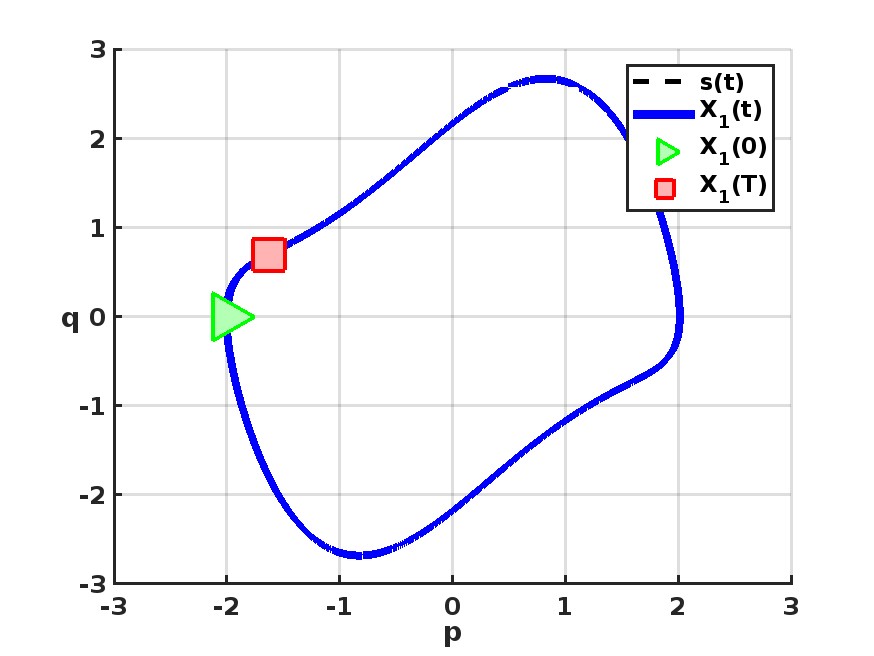}
            \caption{Node 1, $x_{1}(0) = (-2,0)$}
        \end{subfigure}
        \begin{subfigure}{0.22\textwidth}
            \includegraphics[width=\textwidth]{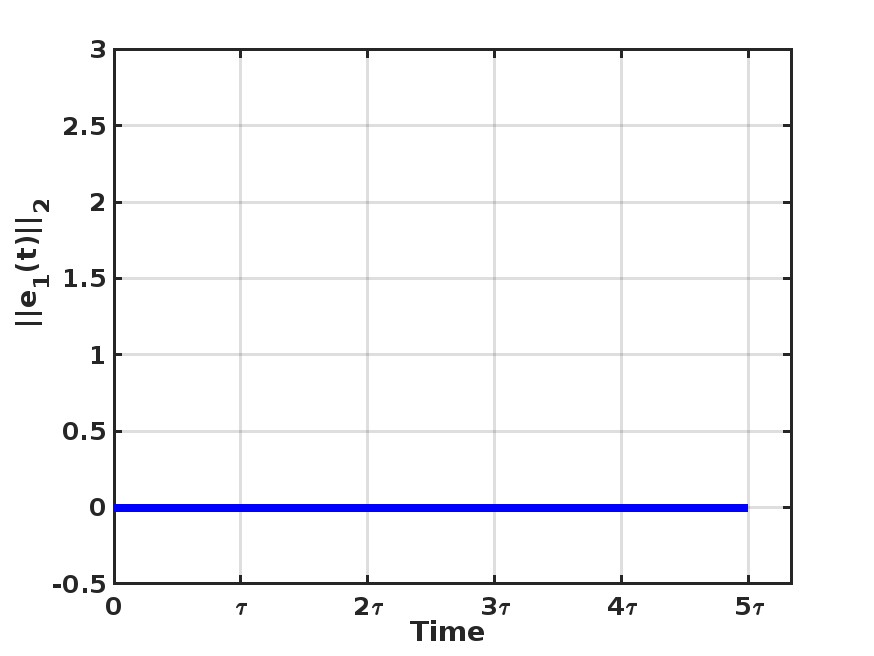}
            \caption{$e_{1}(t)=x_{1}(t)-s(t)$}
        \end{subfigure}
        \hfill 
        \begin{subfigure}{0.22\textwidth}
            \includegraphics[width=\textwidth]{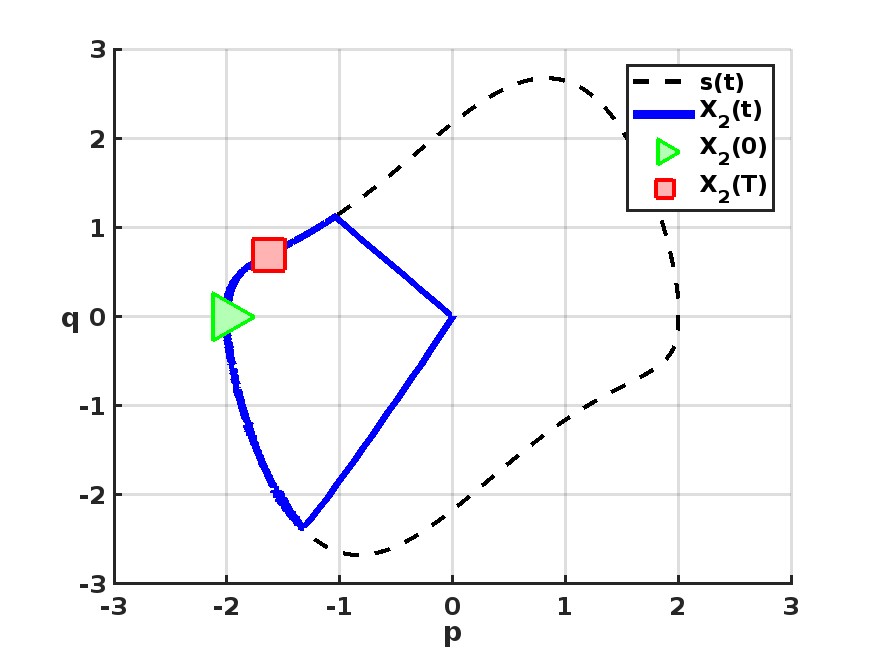}
            \caption{Node 2, $x_{2}(0) = (-2,0)$}
        \end{subfigure}
        \begin{subfigure}{0.22\textwidth}
            \includegraphics[width=\textwidth]{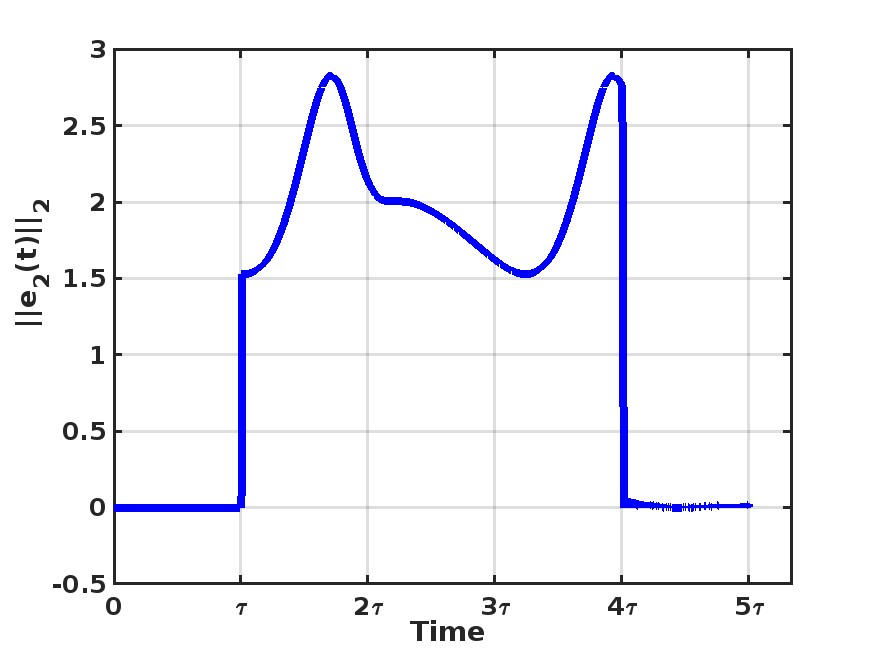}
            \caption{$e_{2}(t)=x_{2}(t)-s(t)$}
        \end{subfigure}        
        \hfill 
        \begin{subfigure}{0.22\textwidth}
            \includegraphics[width=\textwidth]{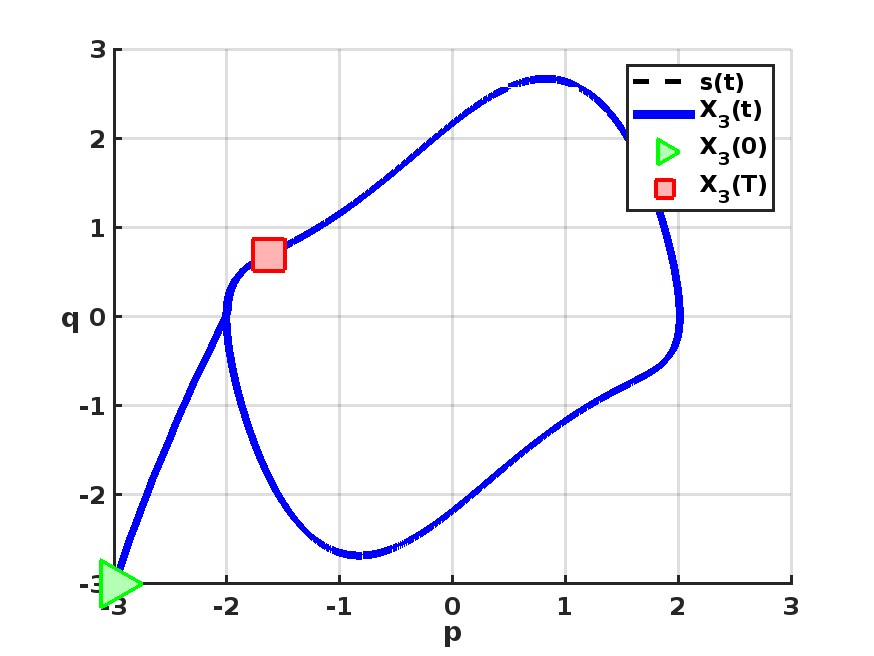}
            \caption{Node 3, $x_{3}(0) = (-3,-3)$}
        \end{subfigure}
        \begin{subfigure}{0.22\textwidth}
            \includegraphics[width=\textwidth]{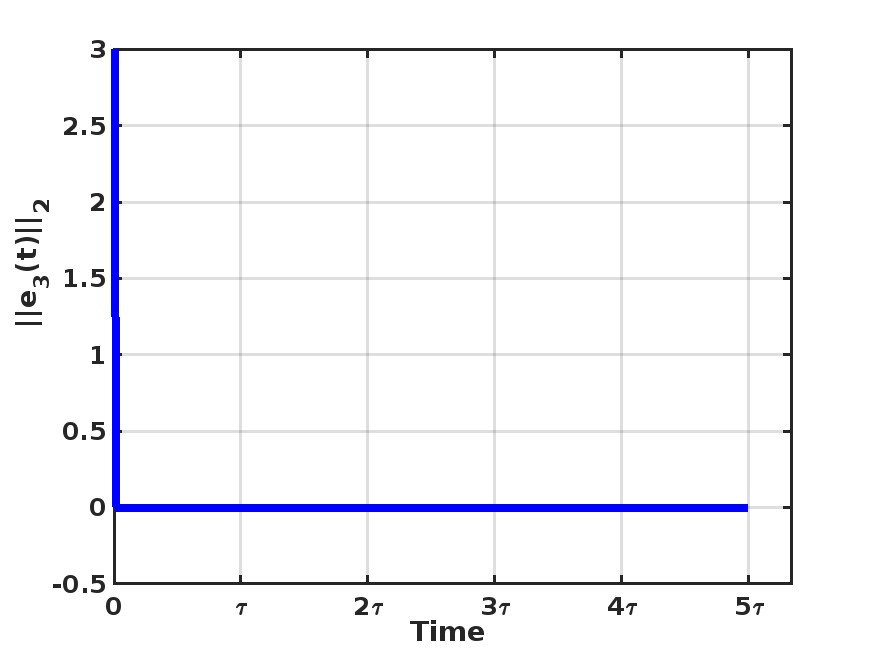}
            \caption{$e_{3}(t)=x_{3}(t)-s(t)$}
        \end{subfigure}
        \hfill 
        \begin{subfigure}{0.22\textwidth}
            \includegraphics[width=\textwidth]{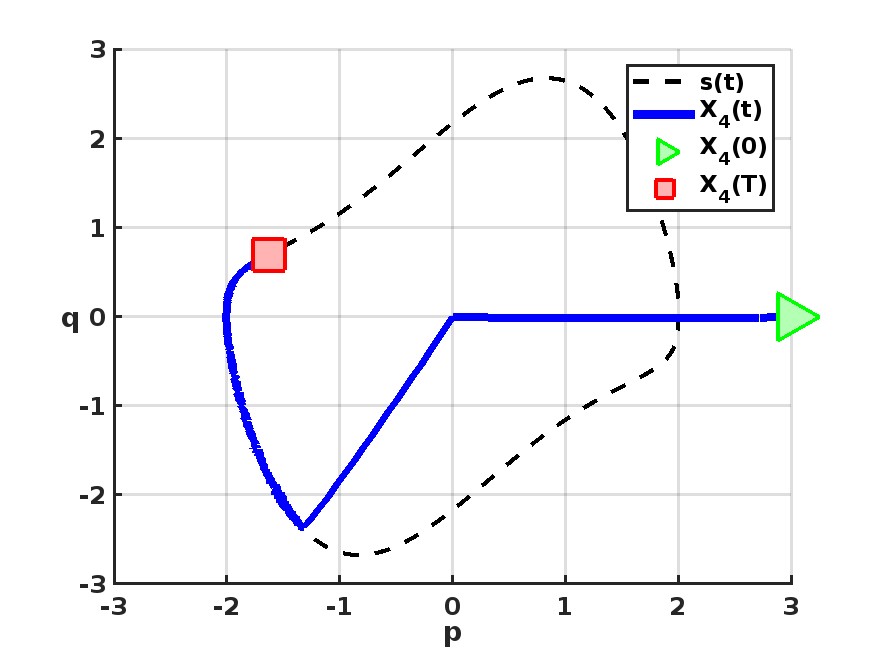}
            \caption{Node 4, $x_{4}(0) = (3,0)$}
        \end{subfigure}
        \begin{subfigure}{0.22\textwidth}
            \includegraphics[width=\textwidth]{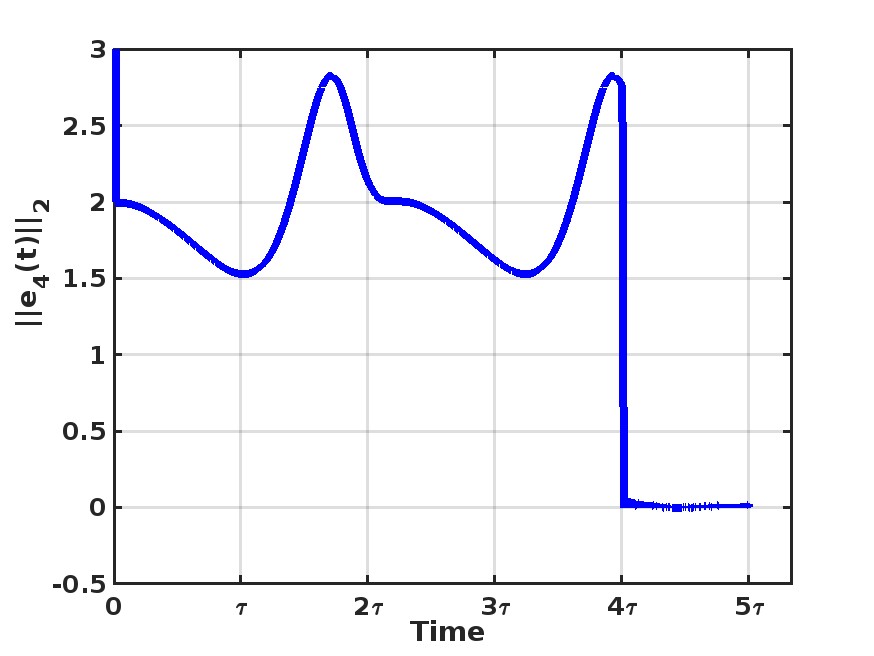}
            \caption{$e_{4}(t)=x_{4}(t)-s(t)$}
        \end{subfigure}
        \hfill
        \centering
        \begin{subfigure}{0.22\textwidth}
            \vspace{10pt}
            \includegraphics[width=\textwidth]{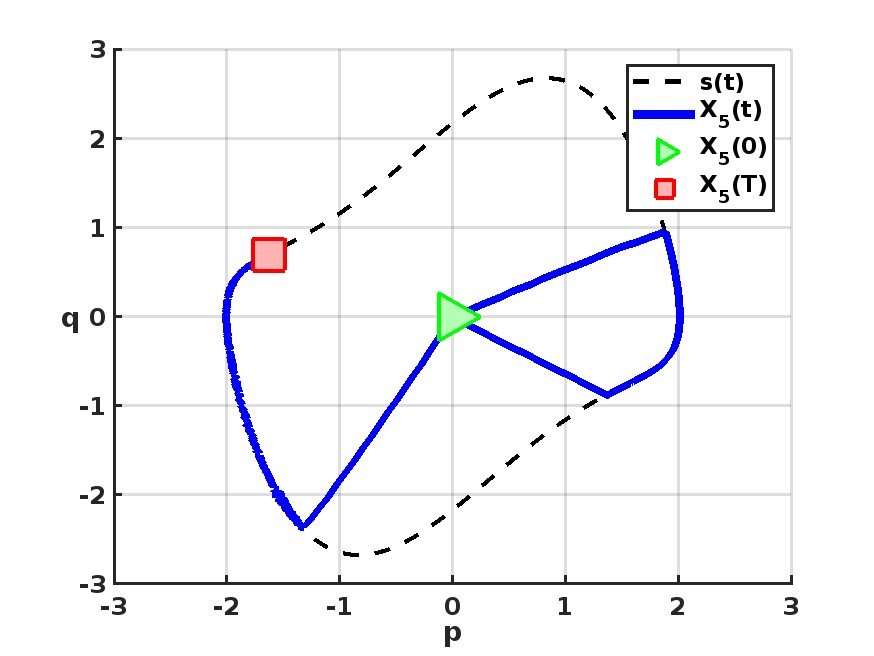}
            \caption{Node 5, $x_{5}(0) = (0,0)$}
        \end{subfigure}
        \begin{subfigure}{0.22\textwidth}
            \includegraphics[width=\textwidth]{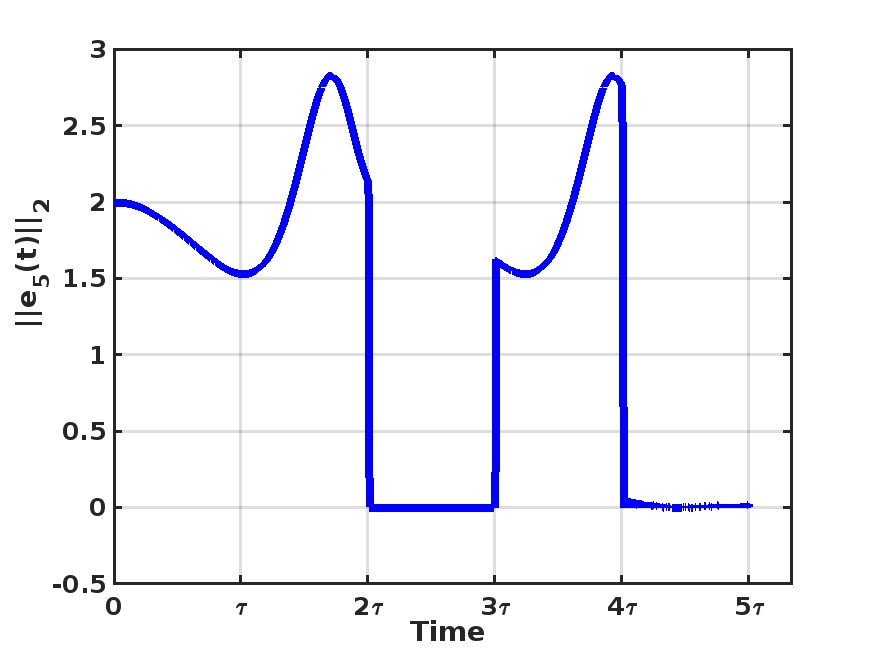}
            \caption{$e_{5}(t)=x_{5}(t)-s(t)$}
        \end{subfigure}
        \hfill    
        \caption{Synchronization of a 5 Node Temporal Network }
        \label{fig:Van der Pol Network Plot}
    \end{figure}

\section{Conclusion}
In this paper we addressed the pinning control synchronization of a network of dynamical systems with time varying interactions modelled using the temporal network framework.  We derived a sufficient condition to synchronize  all the systems in the network to some desired trajectory for a given set of pinning nodes.
We addressed the problem of optimizing the number of pinning nodes to synchronize the network by formulating it as a Linear programming problem.  We next addressed the problem of maximizing the number of synchronized nodes when there are constraints on the number of nodes that could be pinned. We showed that this is a special case of submodular cost submodular knapsack (SCSK) problem which is NP-hard. We proposed a greedy heuristic which gave near optimal solutions in simulation results. In future, we will extend the analysis to temporal networks with directed cycles and networks with non-identical dynamics.
\normalem
\bibliographystyle{ieeetr}
\bibliography{References}
\section{Appendix}\label{sec:appendix}
\subsection{Finite time synchronization for pinning nodes}
An agent/node that is a pinning node will have the following dynamics. For $i\in S_0$, 
\begin{equation}
    \dot{x}_{i}(t) = f(x_{i}(t))+\kappa u(t),
    \label{eq:sysmodel - pinned}
\end{equation}
where $\kappa$ is input gain and 
\begin{equation*}
     u(t) = \begin{cases}
     \dfrac{x(t)-s(t)}{||x(t)-s(t)||_{2}}, &  x(t)\neq s(t)\\
     \mathbb{0}_{n}, & x(t)=s(t) \\ 
     \end{cases}
 \end{equation*}

The above control input will ensure that the state in \eqref{eq:sysmodel - pinned} will be synchronized  to $s(t)$ in finite time for a suitable choice of input gain $\kappa$. This input is inspired from the sliding mode literature \cite{yu2017sliding,edwards1998sliding,shtessel2014sliding}. 

 \subsection{Finite time synchronization in each snapshot of temporal network}
 As discussed in Section \ref{sec:prob_form}, each snapshot $k$ exists for the duration $(k\tau,(k+1)\tau]$. Once the pinning nodes are synchronized, all the other systems interact with each other in order that the entire network is synchronized.
 Consider the dynamics of each agent as given in \eqref{eq:sysmodel - unpinned}. 
We proceed to derive the conditions on $c$ required for finite time synchronization.  Let $e_{i}=x_{i}-s$ be the difference between the states of the node $i$ and the desired trajectory $s$. We choose a quadratic lyapunov function for each node $i$ for an interval $\theta$ as follows:
\[V_{i} = \dfrac{1}{2}e_{i}^{T}e_{i} = \dfrac{1}{2}||e_{i}||^{2}_{2}\]
Taking the time derivative we get 
\begin{align*}
\dot{V}_{i} & = \dfrac{d}{dt}(V_{i}) =\dfrac{\partial V_{i}}{\partial e_{i}}^{T}\dfrac{de_{i}}{dt}=e_{i}^{T}(\dot{x}_{i}-\dot{s})&\\
& = e_{i}^{T}\Big(f(x_{i})-f(s)-cg^{k}_{i}(X)\Big)&\\
\dot{V}_{i} & = e_{i}^{T}\Big(f(x_{i})-f(s)\Big)-\begin{cases}
    ce_{i}^{T}g^{k}_{i}(X)&,x_{i}\neq s\\
    0&,x_{i}=s
\end{cases}&\label{eq:V_dot equation}\\
&\text{If all the incoming edges are from synchronized nodes}&\\
&\text{(i.e) $e_{j}=0$ $\forall j\neq i,$}&\\
\dot{V}_{i} & = e_{i}^{T}\Big(f(x_{i})-f(s)\Big)-\begin{cases}
    ce_{i}^{T}\dfrac{L_{ii}(k-1)e_{i}}{||L_{ii}(k-1)e_{i}||_{2}}&,x_{i}\neq s\\
    0&,x_{i}=s
\end{cases}&\\
 & \dot{V}_{i} = e_{i}^{T}\Big(f(x_{i})-f(s)\Big)-
    \begin{cases}
    c||e_{i}||_{2}&,x_{i}\neq s\\
    0&,x_{i}=s
    \end{cases}&\\
    &\text{Using the QUAD assumption \eqref{eq:QUAD Assumption},}&\\
    \dot{V} & \leq e_{i}^{T}\left(\Delta -\omega I_{n}\right)e_{i}-\begin{cases}
    c||e_{i}||_{2}&,x_{i} \neq s\\
    0&,x_{i}=s
\end{cases}&\\
&\text{Let }Q=\Delta -\omega I_{n}\\
     & \leq \lambda_{max}(Q)||e_{i}||_{2}^{2}-\begin{cases}
    c||e_{i}||_{2}&,x_{i}\neq s\\
    0&,x_{i}=s
\end{cases}&\\
    \dot{V}_{i} & \leq \begin{cases}
    ||e_{i}||_{2}\Big(||e_{i}||_{2}\lambda_{max}(Q)-c\Big)&,x_{i}\neq s\\
    0&,x_{i}=s
\end{cases}&\\
\dot{V}_{i} & \leq \begin{cases} 
\sqrt{2V_{i}}\Big(\lambda_{max}(Q)\sqrt{2V_{i}}-c\Big)&,x_{i}\neq s\\
0&,x_{i}=s
\end{cases}&
\end{align*}
For $x_{i}\neq s$,
\begin{equation}
    \dot{V}_{i}\leq \sqrt{2V_{i}}\Big(\lambda_{max}(Q)\sqrt{2V_{i}}-c\Big)
    \label{eq:Synchronization Condition clause 1}
\end{equation} 
For $x_{i}=s$,
\begin{equation}
    \dot{V}_{i}=0
    \label{eq:Synchronization Condition clause 2}
\end{equation} 
Rearranging the terms in \eqref{eq:Synchronization Condition clause 1} and integrating both sides,
\begin{align}
    &\left[\dfrac{\ln\left( c - \lambda_{max}(Q)\sqrt{2V_{i}(t)}\right)}{\lambda_{max}(Q)}\right]^{V_{i}(\theta)}_{V_{i}(0)} \leq \theta \nonumber\\
    &\text{Let }c > \lambda_{max}(Q)\sqrt{2V_{i}(0)}\nonumber\\
    &\ln\left( \dfrac{c - \lambda_{max}(Q)\sqrt{2V_{i}(\theta)}}{c - \lambda_{max}(Q)\sqrt{2V_{i}(0)}}\right ) \leq \lambda_{max}(Q)\theta\nonumber\\
    \end{align}
 
    \begin{align}
    &||e_{i}(\theta)||_{2}\geq \dfrac{c -  \left(c - \lambda_{max}(Q)||e_{i}(0)||_{2}\right)\boldsymbol{e}^{\lambda_{max}(Q)\theta}}{\lambda_{max}(Q)}
\end{align}
For the node to be synchronized within $\theta$, $||e_{i}(\theta)||_{2} = 0$.
\begin{align}
    \implies & \left(1-\boldsymbol{e}^{-\lambda_{max}(Q)\theta} \right)c\geq \lambda_{max}(Q)||e_{i}((k-1)\theta))||_{2}  \nonumber\\
    &\text{Since } \lambda_{max}(Q)\theta>0 \text{ ,}\nonumber\\
    & c \geq \dfrac{\lambda_{max}(Q)||e_{i}(0)||_{2}}{\left(1 - \boldsymbol{e}^{-\lambda_{max}Q\theta}\right)}\label{eq:Sufficient Coupling Strength - Node - 1}
\end{align}
We observe that the coupling strength $c$ depends on the initial error, system dynamics which determine $\lambda_{max}(Q)$ and $\theta$ the duration within which the finite time synchronization should be achieved. A node could be synchronized at the last snapshot of the temporal network also. Hence, the initial error of the node when it is synchronized is not known.
Therefore using equation \eqref{eq:Sufficient Coupling Strength - Node - 1} as a guideline, we choose a high enough value for $c$ to ensure finite time synchronization.
\balance
Note in each snapshot, the synchronization starts from the root nodes of $G_k$ and then depending on the graph structure and synchronization status of root nodes, each node subsequently gets synchronized if possible. Hence for a leaf node in $G_k$ to get synchronized we need time at most $D_k\theta$, where $D_k$ is the diameter of the graph $G_k$. The duration of each snapshot, $\tau$ depends on the application and hence we assume $\theta$, the finite time taken for a node to get synchronized  satisfies the following condition: $\theta \leq \frac{\tau}{D_k}$
Note that this an upper bound and depending on the graph structure, one can choose better values for $\theta$ with the knowledge of the graph and initial conditions.

\end{document}